\documentclass[letterpaper]{IEEEtran}

\IEEEoverridecommandlockouts

\usepackage{url}
\usepackage[compress]{cite}
\usepackage{amssymb}
\usepackage{amsthm}
\usepackage{dsfont}
\usepackage{mathtools}
\usepackage{amsmath}
\usepackage{xfrac}
\usepackage{stmaryrd}

\newtheorem{theorem}{Theorem}

\newtheorem{definition}{Definition}
\newtheorem{proposition}[theorem]{Proposition}

\newtheorem{remark}{Remark}

\newtheorem{corollary}[theorem]{Corollary}

\DeclareMathOperator*{\argmin}{arg\,min}

\newcommand{\range}[1]{\llbracket #1 \rrbracket}

\begin{document}
\title{Non-Stochastic Private Function Evaluation} 

\author{Farhad Farokhi and Girish Nair
	\thanks{The authors are with the Department of Electrical and Electronic Engineering at the University of Melbourne. e-mails:\{ffarokhi,gnair\}@unimelb.edu.au }
	\thanks{The work of F. Farokhi is supported by an startup grant at Melbourne School of Engineering at the University of Melbourne. The work of G. Nair was supported by the Australian Research Council  grant FT140100527.}
}

\maketitle

\begin{abstract} We consider private function evaluation to provide query responses based on private data of multiple untrusted entities in such a way that each cannot learn something substantially new about the data of others. First, we introduce perfect non-stochastic privacy in a two-party scenario. Perfect privacy amounts to conditional unrelatedness of the query response and the private uncertain variable of other individuals conditioned on the uncertain variable of a given entity. We show that perfect privacy can be achieved for queries that are functions of the common uncertain variable, a generalization of the common random variable. We compute the closest approximation of the queries that do not take this form. To provide a trade-off between privacy and utility, we relax the notion of perfect privacy. We define almost perfect privacy and show that this new definition equates to using conditional disassociation instead of conditional unrelatedness in the definition of perfect privacy. Then, we generalize the definitions to multi-party function evaluation (more than two data entities). We prove that uniform quantization of query responses, where the quantization resolution is a function of privacy budget and sensitivity of the query (cf., differential privacy), achieves function evaluation privacy. 
\end{abstract}

\begin{IEEEkeywords}
Non-Stochastic Privacy; Common Uncertain Variable; Information Leakage; Quantization.
\end{IEEEkeywords}

\section{Introduction}
Privacy research in information theory~\cite{du2012privacy,issa2016operational} and computer science~\cite{dwork2006calibrating,10100797835407922841} often deals with the problem of reporting privacy-preserving query responses based on private data available on a secure server. That is, when computing the privacy-preserving responses, the server has access to the entire private dataset and generates a noisy output with desired utility and privacy against a third-party adversary. Those studies fail to investigate private information leakage from the query responses to individuals whose data is used for responding to the query, which amounts to privacy analysis in the presence of side-channel information.  This is an important problem when multiple untrusted parties must come together to compute the response to an aggregate query.

In this work, we consider providing query responses based on the data of multiple untrusted entities in such a way that they cannot learn something substantially new about the data of others. We refer to this as \textit{private function evaluation}. At first, we introduce and investigate perfect privacy in a two-party scenario. This is  the non-stochastic counterpart of perfect privacy in stochastic literature~\cite{miklau2007formal,calmon2015fundamental, rassouli2018perfect}. Perfect privacy equates to conditional unrelatedness of the query response and the private uncertain variable of other individuals conditioned on the uncertain variable of a given entity. We show that perfect privacy can be achieved for queries that are functions of the common uncertain variable, a generalization of the common random variable in the sense of~\cite{wolf2004zero}. For queries that do not take this form, we can approximate the query to compute the response with the most utility, i.e., closest worst-case response. To provide a trade-off between privacy and utility, we relax the notion of privacy. We define almost perfect non-stochastic privacy and show that this new definition equates to using conditional disassociation, borrowed from~\cite{rangi2019towards}, instead of conditional unrelatedness in the definition of perfect privacy. We investigate the family of functions that can achieve almost perfect privacy. Then, we generalize our definition to multi-party function evaluation (more than two data entities). We prove that private function evaluation can be achieved by uniform quantization of the query responses, where the quantization resolution is a function of privacy budget and sensitivity of the query (cf., scale of the Laplace mechanism differential privacy~\cite{10100797835407922841}). We also investigate the utility of these reporting policies to establish utility-privacy trade-off.

Private function evaluation, at first glance, might seem related to private computation~\cite{kim2019private,8849287, tahmasebi2019private} and private information retrieval~\cite{lin2018asymmetry,banawan2018private}. However, in private function evaluation, the objective is not to hide the function to be computed (cf., private computation) or the datasets on which the function is evaluated (cf., private information retrieval). We are rather interested in ensuring  that the parties contributing to the private dataset cannot infer the private information of the other parties. This links the problem more intimately to privacy-preserving distributed decision making~\cite{chaudhuri2009privacy,abadi2016deep, wu2019value,wainwright2012privacy}. However, all those papers consider stochastic mechanisms for ensuring privacy while, in this paper, we are interested in a non-stochastic notion of information leakage and privacy. 

Despite their shortcomings, due to heuristic-based development, non-stochastic privacy-preserving policies have remained popular~\cite{samarati2001protecting,sweeney2002k,1617392}.
Those studies are motivated by concerns about the use of randomization in popular stochastic approaches. For instance, randomized policies, stemming from differential privacy in financial auditing, can potentially complicate fraud detection~\cite{bhaskar2011noiseless, nabar2006towards} and cause difficulties in medical, health, or social research~\cite{dankar2013practicing,Mervis114}. This has motivated the use of information-theoretic tools to investigate non-stochastic privacy recently~\cite{farokhinoiselessprivacy,8662687,ding2019developing,Nonstochastichypothesis2020}. Nonetheless, those studies do not consider the private function computation setup in this paper.  

\section{Uncertain Variables}
In this section, we present necessary preliminaries from non-stochastic information theory from~\cite{nair2013nonstochastic,8662687}.

A sample space $\Omega$ models the set of uncertainties. An uncertain variable is a mapping from the sample space to a set of interest, such as $X:\Omega\rightarrow\mathbb{X}$. Here, $X(\omega)$ denotes a realization of uncertain variable $X$ corresponding to sample $\omega\in\Omega$. If $\mathbb{X}$ is finite, the uncertain variable $X$ is discrete. In this paper, we focus on discrete uncertain variables. The range of an  uncertain variable $X$ is the set of all its realizations, e.g.,  $\range{X}:=\{X(\omega):\omega\in\Omega \}\subseteq \mathbb{X}.$ The joint range of any two uncertain variables $X:\Omega\rightarrow\mathbb{X}$ and $Y:\Omega\rightarrow\mathbb{Y}$ is $\range{X,Y}:=\{(X(\omega), Y(\omega)):\omega\in\Omega \}\subseteq \mathbb{X}\times \mathbb{Y}.$ The conditional range of uncertain variable $X$, conditioned on a realization of another uncertain variable $Y(\omega)=y$, is  $\range{X|Y(\omega)=y}:=\{X(\omega):\exists \omega\in\Omega  \mbox{ such that } Y(\omega)=y\}=X(Y^{-1}(y))\subseteq \range{X}$, where $Y^{-1}$ denotes the pre-image or inverse image of $Y$. Whenever it is evident from the context, $\range{X|y}$ is used instead of $\range{X|Y(\omega)=y}$ for the sake of brevity.

Uncertain variables $X_1$, $\dots$, $X_n$ are unrelated if $\range{X_1,\dots,X_n}=\range{X_1}\times \cdots \times \range{X_n}.$ Further, they are  conditionally unrelated (conditioned on observations of $Y$) if 
$\range{X_1,\dots,X_n|Y(\omega)=y}=\range{X_1|Y(\omega)=y}\times \cdots \times \range{X_n|Y(\omega)=y}$ for all $y\in\range{Y}$.
For two uncertain variables, this definition is equivalent to stating that $X_1$ and $X_2$ are unrelated if $\range{X_1|X_2(\omega)=x_2}=\range{X_1}, \forall x_2\in\range{X_2},$ and \textit{vice versa}~\cite{nair2013nonstochastic}. 
Again, for two uncertain variables, $X_1$ and $X_2$ are conditionally unrelated (conditioned on realizations of $Y$) if
$\range{X_1|X_2(\omega)=x_2,Y(\omega)=y}=\range{X_1|Y(\omega)=y},\forall (x_2,y)\in\range{X_2,Y}$~\cite{nair2013nonstochastic}.
 
The non-stochastic entropy of uncertain variable $X$ is 
\begin{align} \label{eqn:firstentropys}
H_0(X):=\log_2(|\range{X}|).
\end{align}
This is commonly referred to as the Hartley entropy~\cite{hartley1928transmission,nair2013nonstochastic} and coincides with the R\'{e}nyi entropy of order $0$ for discrete variables, see, e.g.,~\cite[Definition~1]{sason2017arimoto} and~\cite{renyi1961measures}. Conditional (or relative) entropy of uncertain variable $X$ given uncertain variable $Y$ is
\begin{align} \label{eqn:relativeentropys}
H_0(X|Y):=\max_{y\in\range{Y}}\log_2(|\range{X|Y(\omega)=y}|).
\end{align}
This coincides with the Arimoto-R\'{e}nyi conditional entropy of order $0$, see, e.g.,~\cite[Definition~3]{sason2017arimoto} and~\cite{10022581674}. So, the non-stochastic information between two uncertain variables $X$ and $Y$ can be defined as the difference of the entropy of $X$ with and without access to realizations of $Y$:
\begin{align}
I_0(X;Y):=&H_0(X)-H_0(X|Y)\nonumber\\
=&\min_{y\in \range{Y} } \log_2\left(\frac{|\range{X}|}{|\range{X|y}|} \right).\label{eqn:information1}
\end{align} 
Zeroth order information, defined above, is not symmetric, i.e., $I_0(X;Y)\neq I_0(Y;X)$, in general. This notion of information is also related to Kolmogorov's information gain $|\range{X}|/|\range{X|y}|$ and the `combinatorial' conditional entropy $\log_2(|\range{X|y}|)$~\cite{kolmogorov1959varepsilon}. However, the combinatorial conditional entropy and the information gain are defined for a given realization $Y(\omega) = y$ while the non-stochastic information is for the worst-case realization. 

In~\cite{8662687}, it was observed that,  in the context of information-theoretic privacy, $I_0$ is not an appropriate measure of information leakage. This is because $I_0$ focuses on least informative observations while a privacy-intrusive adversary is interested in most informative realizations. Therefore, in~\cite{8662687}, an alternative non-stochastic information leakage is proposed as
\begin{align}
L_0(X;Y):=&\max_{y\in \range{Y} } \log_2\left(\frac{|\range{X}|}{|\range{X|y}|} \right).\label{eqn:information2}
\end{align}
Non-stochastic information leakage $L_0(X;Y)$ captures the worst-case reduction in the complexity of brute-force guessing  $X$ after observing $Y$~\cite{farokhi2020measuring}. In general, $I_0$ and non-stochastic information leakage $L_0$ are not equal, i.e.,  $I_0(X;Y)\neq L_0(Y;X)$. In fact, it is evident that $I_0(X;Y)\leq L_0(X;Y)$. Again, $L_0(X;Y)$ is not symmetric. Note that $L_0(X;Y)\geq 0$ with equality achieved if and only if $X$ and $Y$ are unrelated. Finally, we can similarly define conditional non-stochastic information leakage:
\begin{align}
L_0(X&;Y|Z)\nonumber\\
:=&\max_{(y,z)\in \range{Y,Z} } \log_2\left(\frac{|\range{X|Z(\omega)=z}|}{|\range{X|Y(\omega)=y,Z(\omega)=z}|} \right).
\label{eqn:conditional_non_stochastic_information_leakage}
\end{align}
Note that $L_0(Y;X|Z)\geq 0$ with equality achieved if and only if $X$ and $Y$ are unrelated conditioned on $Z$ .

In~\cite{nair2013nonstochastic}, maximin or non-stochastic information is introduced as a symmetric measure of information and its relationship with zero-error capacity is explored. To present the definition of the maximin information, we need to introduce overlap partitions. 

\begin{definition}[Overlap Partition]\hfill\break\vspace{-1em}
	\begin{itemize}
		\item $x,x'\in\range{X}$ are overlap connected (via $\range{X|Y}$), $x\leftrightsquigarrow x'$,  if there exists a finite sequence of conditional ranges $\{\range{X|y_i}\}_{i=1}^n$ such that $x\in\range{X|y_1}$, $x'\in\range{X|y_n}$, and $\range{X|y_i}\cap \range{X|y_{i+1}}\neq \emptyset$ for all $i=1,\dots,n-1$;
		\item $\mathcal{A}\subseteq\range{X}$ is overlap connected if all $x,x'\in\mathcal{A}$ are overlap connected;
		\item $\mathcal{A},\mathcal{B}\subseteq\range{X}$ are overlap isolated if there do not exist points $x\in\mathcal{A}$ and $x'\in\mathcal{B}$ such that $x\leftrightsquigarrow x'$;
		\item An overlap partition of $\range{X}$ is a set of sets $\range{X|Y}_\star:=\{\mathcal{A}_i\}_{i=1}^n$ such that $\range{X}\subseteq\bigcup_{i=1}^n\mathcal{A}_i$, $\mathcal{A}_i,\mathcal{A}_j$ are  overlap isolated if $j\neq i$, and $\mathcal{A}_i$ is overlap connected;
	\end{itemize}
\end{definition}

There always exists a unique overlap partition~\cite{nair2013nonstochastic}. The maximin information is
\begin{align}
I_\star(X;Y):=\log_2(|\range{X|Y}_\star|).
\end{align}
Note that $I_\star(X;Y)\geq 0$ and $I_\star(X;Y)=0$ if and only if uncertain variables $X$ and $Y$ are unrelated.

\begin{definition}[Taxicab Connectivity]\hfill\break\vspace{-1em} \label{def:taxicab}
	\begin{itemize}
		\item $(x,y),(x',y')\in\range{X,Y}$ are taxicab connected if there exists a sequence of points $\{(x_i,y_i)\}_{i=1}^n\subseteq\range{X,Y}$ such that $(x_1,y_1)=(x,y)$, $(x_n,y_n)=(x',y')$, and either $x_i=x_{i-1}$ or $y_i=y_{i-1}$ for all $i\in\{2,\dots,n\}$;
		\item $\mathcal{A}\subseteq\range{X,Y}$ is taxicab connected if all points in $\range{X,Y}$ are taxicab connected;
		\item $\mathcal{A},\mathcal{B}\subseteq\range{X,Y}$ are taxicab isolated if there do not exist points $(x,y)\in\mathcal{A}$ and $(x',y')\in\mathcal{B}$ such that $(x,y)$ and $(x',y')$ are taxicab connected;
		\item A taxicab partition of $\range{X,Y}$ is a set of sets $\mathfrak{T}(X;Y):=\{\mathcal{A}_i\}_{i=1}^n$ such that $\range{X,Y}\subseteq\bigcup_{i=1}^n\mathcal{A}_i$, any  $\mathcal{A}_i,\mathcal{A}_j$ are  taxicab isolated if $j\neq i$, and $\mathcal{A}_i$ is taxicab connected.
	\end{itemize}
\end{definition}
Again,there exists a unique taxicab partition $\mathfrak{T}(X;Y)$~\cite{nair2013nonstochastic}. Furthermore, $\range{X|Y}_\star$ are $\range{Y|X}_\star$ are projections of the unique taxicab partition $\mathfrak{T}(X;Y)$. Hence, $|\range{X|Y}_\star|=|\mathfrak{T}(X;Y)|=|\range{Y|X}_\star|$ implying that the maximin information is symmetric, i.e.,  $I_\star(X;Y)=I_\star(Y;X)$. The maximin information is related to the non-stochastic information leakage $I_\star(X;Y)\leq L_0(X;Y)$~\cite{farokhinoiselessprivacy,9116825}. Due to symmetry of maximin information,  $I_\star(X;Y)=I_\star(Y;X)\leq L_0(Y;X)$. 

\section{Common Uncertain Variable, Information, and Perfect Privacy}
In this section, we first discuss extension of common random variables in~\cite{wolf2004zero} to uncertain variables in line with the approach of~\cite{nair2013nonstochastic,mahajanrelationship} for developing privacy-preserving policies. This extends the use of common information~\cite{shannon1953lattice}, also known as the G{\'a}cs-K{\"o}rner common information~\cite{GacsKorner1973}, in perfect privacy~\cite{asoodeh2014notes} to the non-stochastic framework. 

\subsection{Common Uncertain Variable and Information}

We start by introducing the notion of common uncertain variables and relating it to overlap partitions and maximin information. 

\begin{definition}[Common Uncertain Variable] 
	\label{def:common}
Let $X_1$ and $X_2$ be any two uncertain variables with disjoint\footnote{The disjoint assumption is just to simplify definition of the bipartite graph by making vertexes associated with alphabet of $X_1$ and $X_2$ distinguishable. This is clearly without loss of generality as changing the event sets/alphabets of uncertain variables does not change their properties.} ranges. 
\begin{itemize}
\item $\mathcal{G}$ is a bipartite graph with the vertex set $\mathcal{V}=\range{X_1}\cup\range{X_2}$ and the edge set $\mathcal{E}=\range{X_1,X_2}$;
\item $f_1:\range{X_1}\rightarrow 2^{\range{X_1}\cup\range{X_2}}$ is a function that maps $x_1\in\range{X_1}\subseteq \mathcal{V}$ to the set of vertices in $\range{X_1}\cup\range{X_2}$ that are in the connected component of $\mathcal{G}$ containing $x_1$;
\item $f_2:\range{X_2}\rightarrow 2^{\range{X_1}\cup\range{X_2}}$ is a function that maps $x_2\in\range{X_2}\subseteq \mathcal{V}$ to the set of vertices in $\range{X_1}\cup\range{X_2}$ that are in the connected component of $\mathcal{G}$ containing $x_2$.
\end{itemize}
The common uncertain variable is $X_1\wedge X_2=f_1\circ X_1=f_2\circ X_2$.
\end{definition}

Similar to~\cite{wolf2004zero}, we should note that the common uncertain variable $X_1\wedge X_2$ is the ``largest'' uncertain variable that can be extracted from uncertain variables $X_1$ and $X_2$. 

\begin{proposition} \label{prop:largest_uv} Assume that uncertain variables $X_1$, $X_2$, and $C$ exist such that $C=\bar{f}_1\circ X_1=\bar{f}_2\circ X_2$ for functions $\bar{f}_1:\range{X_1}\rightarrow\range{C}$ and $\bar{f}_1:\range{X_2}\rightarrow\range{C}$. There exists $g:2^{\range{X_1}\cup\range{X_2}}\rightarrow \range{C}$ such that $C=g(X_1\wedge X_2)$. 
\end{proposition}

\begin{IEEEproof}
The proof follows the same line of reasoning as in the proof of Lemma 1 in~\cite{wolf2004zero}.
\end{IEEEproof}

For the unique overlap partition of $\range{X_1}$, $\range{X_1|X_2}_\star=\{\mathcal{A}_i\}_{i=1}^{n_1}$ with $n_1\in\mathbb{N}$, define $\mathfrak{i}_1:\range{X_1}\rightarrow \{1,\dots,n_1\}$ such that $\mathfrak{i}_1(x_1)=i$ for which $x_1\in\mathcal{A}_i$. Similarly, for the unique overlap partition of  $\range{X_2}$, $\range{X_2|X_1}_\star=\{\mathcal{B}_i\}_{i=1}^{n_2}$ with $n_2\in\mathbb{N}$, define $\mathfrak{i}_2:\range{X_2}\rightarrow \{1,\dots,n_2\}$ such that $\mathfrak{i}_2(x_2)=i$ for which $x_2\in\mathcal{B}_i$. The mappings $\mathfrak{i}_1$ and $\mathfrak{i}_1$ are well-defined because $\range{X_1|X_2}_\star$ and $\range{X_2|X_1}_\star$ are partitions. 

\begin{definition}[Equivalence]
	Two uncertain variables $X$ and $Y$ are equivalent, $X\equiv Y$, if there exists a one-to-one correspondence\footnote{The partitions of sample space $\Omega$ induced by $X$ and $Y$ are the same though their labeling may be different.} $f:\range{X}\rightarrow\range{Y}$ such that $Y=f\circ X$. 
\end{definition}

The notion of ``equivalence” between two uncertain variables is weaker than that of ``equality''. If two uncertain variables are equal, they are equivalent as well. However, if two uncertain variables are equivalent, they could differ in terms of their range and hence not being equal to each other. This concept is explored for random variables in~\cite{li2011connection}. Note that entropy remains invariant under the equivalence relationship. 

\begin{proposition}
$H_0(X)=H_0(Y)$ if $X\equiv Y$.
\end{proposition}

\begin{proof} Since $X\equiv Y$, there must exist a one-to-one correspondence $f$ such that $Y=f\circ X$.  Note that $|\range{Y}|=|\{Y(\omega):\omega\in\Omega \}|=|\{f(X(\omega)):\omega\in\Omega \}|=|\{X(\omega):\omega\in\Omega \}|=|\range{X}|$, where the third equality follows from that $f$ is a one-to-one correspondence.
\end{proof}

\begin{proposition} \label{prop:cuv_maximin}
$X_1\wedge X_2\equiv \mathfrak{i}_1\circ X_1\equiv \mathfrak{i}_2\circ X_2$ and $I_\star(X_1;X_2)=H_0(X_1\wedge X_2)$.
\end{proposition}

\begin{IEEEproof}
Again, $\range{X_1|X_2}_\star$ are $\range{X_2|X_1}_\star$ are projections of the unique taxicab partition  $\mathfrak{T}(X;Y)$~\cite{nair2013nonstochastic}. The elements of the taxicab partition correspond to the connected components of the bipartite graph $\mathcal{G}$~\cite{nair2013nonstochastic,mahajanrelationship}. The rest of the proof follows from that entropy is invariant under equivalence relationship between uncertain variables (similar to the case of random variables in~\cite{li2011connection,griffith2014intersection}).
\end{IEEEproof}

\begin{remark}[Common Knowledge vs Common Information] Under additional measurability assumptions~\cite{mahajanrelationship}, taxicab partition and maximin information also relate to common knowledge in economics~\cite{aumann1976agreeing, aumann1999interactive} and logic~\cite{meyer2004epistemic,van2007dynamic}, not to be mistaken with common information~\cite{shannon1953lattice,GacsKorner1973,wolf2004zero,wyner1975common}. Instead of the graph theoretic definition for common information in~\cite{wolf2004zero} (and its extension to uncertain variables in Definition~\ref{def:common}), common knowledge  of a statement within a group refers to that everyone knows the validity of the aforementioned statement, everyone knows that everyone knows this, everyone knows that everyone knows that everyone knows this, and so on \textit{ad infinitum}. This is modelled and captured using knowledge operators and knowledge field~\cite{aumann1976agreeing}. 
\end{remark}

\begin{remark}[Extension to Multi-Variables]
To define common uncertain variable between multiple (more than two) uncertain variables, the approach of~\cite{tyagi2011function} for multi-variable extension of the G{\'a}cs-K{\"o}rner common information can be used. Let $X_1\dots,X_n$ be any $n\geq 2$ uncertain variables with disjoint ranges. The common uncertain variable is recursively constructed by $X_1\wedge X_2\wedge \cdots \wedge X_{i-1} \wedge X_{i}=(X_1\wedge X_2\wedge \cdots \wedge X_{i-1})\wedge X_i$ for all $i=2,\dots,n$. Note that the binary operation $\wedge$ between uncertain variables is associative~\cite{wolf2004zero} and commutative (see~\cite{li2011connection,griffith2014intersection} for the random variable case).
\end{remark}

\subsection{Perfect Privacy}

Perfect privacy~\cite{miklau2007formal}, defined by adapting Shannon's perfect secrecy~\cite{shannon1949communication} to the privacy framework, states that an observations is perfectly private if it is statistically independent of the secret/private random variable. This concept has been recently further investigated~\cite{calmon2015fundamental, rassouli2018perfect} as it provides a fundamental understanding of utility-privacy trade-off. In the non-stochastic case, independence can be replaced with unrelatedness. We can tailor this definition to the case of private function computation by assuming that, conditioned on the realization of the uncertain variable of each party, the outcome should not leak any information about the uncertain variable of the other party, i.e., the outcome should be conditionally unrelated to each uncertain variable. 

\begin{definition}[Perfect Privacy in Two-Party Function Evaluation] \label{def:perfect} Let $X_1$ and $X_2$ be any two uncertain variables. The mapping $f:\range{X_1,X_2}\rightarrow \mathbb{R}^m$ provides perfect privacy if $f(X_1,X_2)$ is unrelated to $X_1$  conditioned on $X_2$ and $f(X_1,X_2)$ is unrelated to $X_2$  conditioned on $X_1$.
\end{definition}

Note that Definition~\ref{def:perfect} implies that the mapping $f$ provides perfect privacy if $\range{X_2|X_1(\omega)=x_1, Z(\omega)=z} = \range{X_2|X_1(\omega)=x_1}$ and $\range{X_1|X_2(\omega)=x_2, Z(\omega)=z} = \range{X_1|X_2(\omega)=x_2}$ with $Z=f(X_1,X_2)$. Perfect privacy for two-party function evaluation can be equivalently defined using conditional non-stochastic information leakage. This is proved in the next proposition.

\begin{proposition} \label{prop:different_def_perfect} Mapping $f$ provides perfect privacy if and only if $L_0(f(X_1,X_2);X_1|X_2)=L_0(f(X_1,X_2);X_2|X_1)=0$.
\end{proposition}

\begin{proof}
We first prove that $f(X_1,X_2)$ is unrelated to $X_1$  conditioned on $X_2$ if and only if $L_0(f(X_1,X_2);X_1|X_2)=0$. First, if $f(X_1,X_2)$ is unrelated to $X_1$  conditioned on $X_2$, we get  $\range{f(X_1,X_2)|X_1(\omega)=x_1,X_2(\omega)=x_2}=\range{f(X_1,X_2)|X_2(\omega)=x_2}$ for all $(x_1,x_2)\in\range{X_1,X_2}$. Thus, $|\range{f(X_1,X_2)|X_1(\omega)=x_1,X_2(\omega)=x_2}|=|\range{f(X_1,X_2)|X_2(\omega)=x_2}|$ for all $(x_1,x_2)\in\range{X_1,X_2}$. Therefore,  definition of conditional non-stochastic information leakage in~\eqref{eqn:conditional_non_stochastic_information_leakage} implies that $L_0(f(X_1,X_2);X_1|X_2)=0$. Now, we prove the reverse. Assume that $L_0(f(X_1,X_2);X_1|X_2)=0$. Then, it must be that $|\range{f(X_1,X_2)|X_1(\omega)=x_1,X_2(\omega)=x_2}|=|\range{f(X_1,X_2)|X_2(\omega)=x_2}|$ for all $(x_1,x_2)\in\range{X_1,X_2}$. Note that $\range{f(X_1,X_2)|X_1(\omega)=x_1,X_2(\omega)=x_2}\subseteq \range{f(X_1,X_2)|X_2(\omega)=x_2}$. Therefore, $|\range{f(X_1,X_2)|X_1(\omega)=x_1,X_2(\omega)=x_2}|=|\range{f(X_1,X_2)|X_2(\omega)=x_2}|$ implies that $\range{f(X_1,X_2)|X_1(\omega)=x_1,X_2(\omega)=x_2}=\range{f(X_1,X_2)|X_2(\omega)=x_2}$. This shows that $f(X_1,X_2)$ is unrelated to $X_1$  conditioned on $X_2$. Similarly, we can prove that $f(X_1,X_2)$ is unrelated to $X_2$  conditioned on $X_1$ if and only if $L_0(f(X_1,X_2);X_2|X_1)=0$. This concludes the proof.
\end{proof}

Evaluating common uncertain variable $X_1\wedge X_2$ and its functions provide perfect privacy. This is proved in the following proposition. 

\begin{proposition} The following statements hold: 
	\begin{itemize}
		\item $X_1\wedge X_2$ is unrelated to $X_1$ conditioned on $X_2$;
		\item $X_1\wedge X_2$ is unrelated to $X_2$ conditioned on $X_1$. 
	\end{itemize}
\end{proposition}

\begin{IEEEproof} Note that because $X_1\wedge X_2=\mathfrak{i}_2\circ X_2$, we have
$\range{X_1\wedge X_2|X_1(\omega)=x_1,X_2(\omega)=x_2}= \range{\mathfrak{i}_2(X_2)|X_1(\omega)=x_1,X_2(\omega)=x_2}=\{\mathfrak{i}_2(x_2)\}=\range{\mathfrak{i}_2(X_2)|X_2(\omega)=x_2}.$ Therefore, $X_1\wedge X_2$ is unrelated to $X_1$ conditioned on $X_2$. The proof for the other case is similar. 
\end{IEEEproof}

As stated earlier, all functions of the common uncertain variable also provide perfect privacy. The inverse is however also true. In fact, any function that provides perfect privacy must only be computable based on the common uncertain variable. This is explored in the following proposition. 

\begin{proposition}
For any $f:\range{X_1,X_2}\rightarrow \mathbb{R}^m$ providing perfect privacy, there exists $g:2^{\range{X_1}\cup\range{X_2}}\rightarrow \mathbb{R}^m$ such that $f(X_1,X_2)=g(X_1\wedge X_2)$. 
\end{proposition}

\begin{IEEEproof}  Because $f$ is a noiseless deterministic mapping, $\range{f(X_1,X_2)|X_1(\omega)=x_1, X_2(\omega)=x_2}$ is a singleton for all $(x_1,x_2)\in\range{X_1,X_2}$. Noting that $f(X_1,X_2)$ is unrelated to $X_1$ conditioned on $X_2$, $\range{f(X_1,X_2)|X_1(\omega)=x_1,X_2(\omega)=x_2} =\range{f(X_1,X_2)|X_2(\omega)=x_2}$ for all $(x_1,x_2)\in\range{X_1,X_2}$. As a result, $\range{f(X_1,X_2)|X_2(\omega)=x_2}$ is also a singleton for all $(x_1,x_2)\in\range{X_1,X_2}$. Therefore, there exists $g_2$ such that $f(X_1,X_2)=g_2\circ X_2$. Similarly, there exists $g_1$ such that $f(X_1,X_2)=g_1\circ X_1$. Therefore, $g_1\circ X_1=f(X_1,X_2)=g_2\circ X_2$. By Proposition~\ref{prop:largest_uv}, there must exists $g$ such that $f(X_1,X_2)=g(X_1\wedge X_2)$. 
\end{IEEEproof}

Not all functions provides perfect privacy. This is evident as not all functions can be rewritten in terms of the common variable $X_1\wedge X_2$. For instance, $f(X_1,X_2)=X_1$ cannot be written in terms of the common uncertain variable. This function also does not provide perfect privacy. Thus, we might need to approximate such a function with one that provides perfect privacy. This can be done by
\begin{subequations}\label{eqn:approximate}
\begin{align} 
\min_{f'}\quad  & \max_{(x_1,x_2)\in\range{X_1,X_2}} \|f(x_1,x_2)-f'(x_1,x_2)\|,\\
\mathrm{s.t.} \quad  & f' \mbox{ provides perfect privacy}.
\end{align}
\end{subequations}
If $f'$ provides perfect privacy, it must be in the form of $g(X_1\wedge X_2)$ for some mapping $g$. Therefore, in light of Proposition~\ref{prop:cuv_maximin}, we have
\begin{align*}
\max_{(x_1,x_2)\in\range{X_1,X_2}}& \|f(x_1,x_2)-f'(x_1,x_2)\|\\
&=\max_i \max_{(x_1,x_2)\in\mathcal{A}_i} \|f(x_1,x_2)-f'(x_1,x_2)\|,
\end{align*}
where $\mathfrak{T}(X_1;X_2) =\{\mathcal{A}_i\}_{i=1}^{|\mathfrak{T}(X_1;X_2)|}$. The solution of~\eqref{eqn:approximate} is given by
\begin{align*}
f'(x_1,x_2)=
&\mathrm{center}(f(\mathcal{A}_i)),\quad (x_1,x_2)\in\mathcal{A}_i,
\end{align*}
where, for any set $\mathcal{A}\subseteq \mathbb{R}^m$,
\begin{align*}
\mathrm{center}(\mathcal{A})
:=
&\argmin_{z'\in\mathrm{conv}(\mathcal{A})} \max_{z\in \mathcal{A}} \|z'-z\|,
\end{align*}
with $\mathrm{conv}(\mathcal{A})$ denoting the convex hull of $\mathcal{A}$. In general, the condition for perfect privacy can be strong. It also does not offer a systematic way for trading-off utility and privacy.  In the remainder of this paper, we relax this notion of privacy. 

\section{Almost Perfect Privacy}
We can relax the conditional unrelatedness in the definition of perfect privacy in Definition~\ref{def:perfect} to get a weaker notion of privacy.  
Proposition~\ref{prop:different_def_perfect} shows that $\max\{L_0(f(X_1,X_2);X_1|X_2),L_0(f(X_1,X_2);X_2|X_1)\}=0$ if and only if mapping $f$ provides perfect privacy. By definition,  $\max\{L_0(f(X_1,X_2);X_1|X_2),L_0(f(X_1,X_2);X_2|X_1)\}\geq 0$. Therefore, we can relax perfect privacy by requiring that $\max\{L_0(f(X_1,X_2);X_1|X_2),L_0(f(X_1,X_2);X_2|X_1)\}$ is small rather than zero.

\begin{definition}[$\gamma$-Privacy in Two-Party Function Evaluation] \label{def:almost_perfect_new} Let $X_1$ and $X_2$ be any two uncertain variables. For $\gamma\geq 0$, mapping $f:\range{X_1,X_2}\rightarrow \mathbb{R}^m$ provides $\gamma$-privacy if $\max\{L_0(f(X_1,X_2);X_1|X_2),L_0(f(X_1,X_2);X_2|X_1)\}\leq \gamma$. 
\end{definition}

In what follows, we borrow disassociation from~\cite{rangi2019towards} as a relaxation of unrelatedness. This way, we can investigate $\gamma$-privacy in more depth by casting it in terms of disassociation rather than conditional information leakage.

\begin{definition}[Disassociated Uncertain Variables]
For $\delta\in[0,1]$, two uncertain variables $X$ and $Y$ are 
$\delta$-disassociated if 
\begin{subequations}
\begin{align}
\frac{|\range{X|Y(\omega)=y_1}\cap \range{X|Y(\omega)=y_2}|}{|\range{X}|}&\geq \delta, \nonumber\\
&\hspace{-.6in}\forall y_1,y_2\in\range{Y}:y_1\neq y_2,
\label{eqn:delta1}
\\
\frac{|\range{Y|X(\omega)=x_1}\cap \range{Y|X(\omega)=x_2}|}{|\range{Y}|}&\geq \delta, \nonumber\\
&\hspace{-.6in}\forall x_1,x_2\in\range{X}:x_1\neq x_2.
\label{eqn:delta2}
\end{align}
\end{subequations}
If only~\eqref{eqn:delta1} holds, $X$ is partially $\delta$-disassociated with $Y$. 
\end{definition}

\begin{proposition}
	\label{prop:dis_unrelated}
	Two uncertain variables $X$ and $Y$ are unrelated if and only if they are $1$-disassociated.
\end{proposition}

\begin{proof}
The proof is outlined in~\cite{rangi2019towards}. For the proof of the sufficiency, note that $|\range{X|Y(\omega)=y_1}\cap\range{X|Y(\omega)=y_2}| =|\range{X}|$ for all $y_1,y_2\in\range{Y}$ if $\delta=1$ in~\eqref{eqn:delta1}. Noting that, $\range{X|Y(\omega)=y_1}\cap\range{X|Y(\omega)=y_2}\subseteq \range{X}$, $|\range{X|Y(\omega)=y_1}\cap\range{X|Y(\omega)=y_2}| =|\range{X}|$ implies that $\range{X|Y(\omega)=y_1}\cap\range{X|Y(\omega)=y_2}= \range{X}$.
And, in turn, this implies that $\range{X|Y(\omega)=y_1}=\range{X|Y(\omega)=y_2}=\range{X}$ Similarly, $\range{Y|X(\omega)=x_1}= \range{Y|X(\omega)=x_2}=\range{Y}$ for all $x_1,x_2\in\range{X}$ if $\delta=1$ in~\eqref{eqn:delta2}. Therefore, $X$ and $Y$ must be unrelated. The proof of the necessity is similar and follows from straightforward algebraic manipulations. 
\end{proof}

As $\delta$ increases, any two $\delta$-disassociated uncertain variables ``appear more unrelated'' and, as demonstrated in Proposition~\ref{prop:dis_unrelated}, $1$-disassociated implies unrelatedness between two uncertain variables. Therefore, we can think of $\delta$-disassociation as a relaxation of the notion of unrelatedness. 

\begin{proposition} \label{prop:upperbound} The following inequalities hold for $\delta$-disassociated uncertain variables $X$ and $Y$:
	\begin{itemize}
		\item $L_0(X;Y)\leq -\log_2(\delta)$;
		\item $L_0(Y;X)\leq -\log_2(\delta)$;
		\item $I_\star(X;Y)\leq -\log_2(\delta)$.
	\end{itemize}
\end{proposition}

\begin{proof} Let us prove that $L_0(X;Y)\leq -\log_2(\delta)$.  If $|\range{Y}|=1$, $L_0(X;Y)=0$ and the inequality trivially holds. Therefore, without loss of generality, we concentrate on $|\range{Y}|>1$. For any $y\in\range{Y}$ and $y'\neq y$, we get $|\range{X|Y(\omega)=y}|\geq |\range{X|Y(\omega)=y}\cap \range{X|Y(\omega)=y'}|\geq \delta |\range{X}|$. Therefore, $|\range{X}|/|\range{X|Y(\omega)=y}|\leq 1/\delta$. The proof for $L_0(Y;X)\leq -\log_2(\delta)$ follows the same line of reasoning. Finally, the proof can be concluded by noting that $I_\star(X;Y)\leq L_0(X;Y)$~\cite{farokhinoiselessprivacy,9116825}.
\end{proof}

Proposition~\ref{prop:upperbound} shows that the information content between any two $\delta$-disassociated uncertain variables $X$ and $Y$ reduces as $\delta$ gets larger. In the limit for $\delta=1$, $L_0(X;Y)=L_0(Y;X)=I_\star(X;Y)=0$. This also shows that $X$ and $Y$ are unrelated if they are $1$-disassociated.

\begin{definition}[Conditionally Disassociated Uncertain Variables]
	For $\delta\in[0,1]$, two uncertain variables $X$ and $Y$ are $\delta$-disassociated conditioned on uncertain variable $Z$ if 
	\begin{subequations}
		\begin{align}
		&\frac{|\range{X|Y(\omega)\hspace{-.03in}=\hspace{-.03in}y_1,Z(\omega)\hspace{-.03in}=\hspace{-.03in}z}\hspace{-.03in}\cap\hspace{-.03in} \range{X|Y(\omega)\hspace{-.03in}=\hspace{-.03in}y_2,Z(\omega)\hspace{-.03in}=\hspace{-.03in}z}|}{|\range{X|Z(\omega)=z}|}\geq \delta, \nonumber\\
		&\hspace{.4in}\forall y_1,y_2\in\range{Y|Z(\omega)=z}:y_1\neq y_2,\forall z\in\range{Z},
		\label{eqn:delta1_delta}
		\\
		&\frac{|\range{Y|X(\omega)\hspace{-.03in}=\hspace{-.03in}x_1,Z(\omega)\hspace{-.03in}=\hspace{-.03in}z}\hspace{-.03in}\cap\hspace{-.03in} \range{Y|X(\omega)\hspace{-.03in}=\hspace{-.03in}x_2,Z(\omega)\hspace{-.03in}=\hspace{-.03in}z}|}{|\range{Y|Z(\omega)=z}|}\geq \delta, \nonumber\\
		&\hspace{.4in}\forall x_1,x_2\in\range{X|Z(\omega)=z}:x_1\neq x_2,\forall z\in\range{Z}.
		\end{align}
	\end{subequations}
If only~\eqref{eqn:delta1_delta} holds, $X$ is partially $\delta$-disassociated with $Y$ conditioned on  $Z$. 
\end{definition}

Following the same argument as Proposition~\ref{prop:dis_unrelated}, two uncertain variables $X$ and $Y$ are unrelated conditioned on uncertain variable $Z$ if they are $1$-disassociated conditioned on  $Z$. Hence, we can think of conditional disassociation as a relaxation of conditional unrelatedness. 

\begin{proposition} \label{prop:upperbound_conditional} Assume that two uncertain variables $X$ and $Y$ are $\delta$-disassociated conditioned on uncertain variable $Z$ for some $\delta\in[0,1]$. Then, $L_0(X;Y|Z)\leq -\log_2(\delta)$.
\end{proposition}

\begin{proof}
	For any $(y,z)\in\range{Y,Z}$ such that $|\range{Y|Z(\omega)=z}|>1$, we get $|\range{X|Y(\omega)=y,Z(\omega)=z}|\geq |\range{X|Y(\omega)=y,Z(\omega)=z}\cap \range{X|Y(\omega)=y',Z(\omega)=z}|\geq \delta |\range{X|Z(\omega)=z}|$. Therefore, $|\range{X|Z(\omega)=z}|/|\range{X|Y(\omega)=y,Z(\omega)=z}|\leq 1/\delta$. 
For any $(y,z)\in\range{Y,Z}$ such that $|\range{Y|Z(\omega)=z}|=1$, we also get $|\range{X|Y(\omega)=y,Z(\omega)=z}|=1$ because $\range{X|Y(\omega)=y,Z(\omega)=z}\subseteq \range{Y|Z(\omega)=z}$. Therefore, $|\range{X|Z(\omega)=z}|/|\range{X|Y(\omega)=y,Z(\omega)=z}|=1\leq 1/\delta$. This implies that $L_0(X;Y|Z)\leq -\log_2(\delta)$.
\end{proof}

\begin{proposition}  \label{prop:disassociation_relationship} Let $X_1$ and $X_2$ be any two uncertain variables. For $\gamma\geq 0$, mapping $f:\range{X_1,X_2}\rightarrow \mathbb{R}^m$ provides $\gamma$-privacy if $f(X_1,X_2)$ and $X_1$ are $e^{-\gamma}$-disassociated conditioned on $X_2$, and $f(X_1,X_2)$ and $X_2$ are $e^{-\gamma}$-disassociated conditioned on $X_1$.
\end{proposition}

\begin{proof}
The proof follows form the application of Proposition~\ref{prop:upperbound_conditional}.
\end{proof}

Proposition~\ref{prop:disassociation_relationship} shows that we can relax the definition of perfect privacy  by requiring conditional disassociation instead of conditional unrelatedness. 

\begin{definition}[$\delta$-Overlap Connectivity] \label{def:delta_overlap}
	For $\delta\in[0,1]$, 
		\begin{itemize}
		\item $x,x'\in\range{X}$ are $\delta$-overlap connected (via $\range{X|Y}$),  $x\leftrightsquigarrow_\delta x'$, if there exists a sequence of points $\{y_i\}_{i=1}^n\subseteq\range{Y}$ such that $x\in\range{X|Y(\omega)=y_1}$, $x'\in\range{X|Y(\omega)=y_n}$, and $	{|\range{X|Y(\omega)=y_i}\cap \range{X|Y(\omega)=y_{i-1}}|}/{|\range{X}|}
		\geq \delta,$
		for all $i\in\{2,\dots,n\}$;
		\item if $x,x'\in\range{X}$ are $\delta$-overlap connected with $n=1$, they are singly $\delta$-overlap connected;
		\item $\mathcal{A}\subseteq\range{X}$ is (singly) $\delta$-overlap connected if all points in $\range{X}$ are (singly) $\delta$-taxicab connected;
		\item $\mathcal{A},\mathcal{B}\subseteq\range{X}$ are $\delta$-overlap isolated if there do not exist points $x\in\mathcal{A}$ and $x'\in\mathcal{B}$ such that  $x\leftrightsquigarrow_\delta x'$;
		\item A $\delta$-overlap partition of $\range{X}$ is a set of sets $\range{X|Y}_\delta:=\{\mathcal{A}_i\}_{i=1}^n$ such that $\range{X}\subseteq\bigcup_{i=1}^n\mathcal{A}_i$,   $\mathcal{A}_i,\mathcal{A}_j$ are  $\delta$-overlap isolated if $j\neq i$, and $\mathcal{A}_i$ is $\delta$-overlap connected;
		\item A $\delta$-overlap family of $\range{X}$, denoted by $\range{X|Y}_\star^\delta$, is the largest $\delta$-overlap partition of $\range{X}$ such that each set in the family contains a singly $\delta$-overlap connected set of the form $\range{X|Y(\omega)=y}$, there exists a set containing any two singly $\delta$-overlap connected points, and the measure of overlap between any two sets in the family is at most $\delta|\range{X}|$.
	\end{itemize}
\end{definition}

For any two uncertain variables $X$ and $Y$, there always exists a $\delta$-overlap family of $\range{X}$ while the uniqueness is guaranteed if $X$ and $Y$ are $\delta$-disassociated~\cite[Theorems~3-4]{rangi2019towards}. 
For a given $\delta$-overlap family of $\range{X_1}$, $\range{X_1|X_2}_\star^\delta=\{\mathcal{A}_i\}_{i=1}^{n_1}$ with $n_1\in\mathbb{N}$, define $\mathfrak{i}^\delta_1:\range{X_1}\rightarrow \{1,\dots,n_1\}$ such that $\mathfrak{i}^\delta_1(x_1)=i$ for which $x_1\in\mathcal{A}_i$. Similarly, for a given $\delta$-overlap family of $\range{X_2}$, $\range{X_2|X_1}_\star^\delta=\{\mathcal{B}_i\}_{i=1}^{n_2}$ with $n_2\in\mathbb{N}$, define $\mathfrak{i}^\delta_2:\range{X_2}\rightarrow \{1,\dots,n_2\}$ such that $\mathfrak{i}^\delta_2(x_2)=i$ for which $x_2\in\mathcal{B}_i$. The mappings $\mathfrak{i}^\delta_1$ and $\mathfrak{i}^\delta_1$ are well-defined because $\range{X_1|X_2}_\star^\delta$ and $\range{X_2|X_1}_\star^\delta$ partition $\range{X_1}$ and $\range{X_2}$, respectively. 

\begin{proposition} \label{prop:almost_cuv1} If $X_1$ and $X_2$ are $\delta$-disassociated for $\delta\in[0,1]$, $\mathfrak{i}^\delta_1\circ X\equiv \mathfrak{i}^\delta_2\circ X_2$.
\end{proposition}

\begin{proof}
Note that, if $X_1$ and $X_2$ are $\delta$-disassociated, $\range{X_1|X_2}_\star^\delta$ are $\range{X_2|X_1}_\star^\delta$ are unique~\cite[Theorems~3-4]{rangi2019towards} and are projections of the unique $\delta$-taxicab family of $\range{X,Y}$ denoted by $\range{X,Y}_\star^\delta$~\cite[Theorem~5]{rangi2019towards}. Therefore, there is a bijection from $\range{X_1|X_2}_\star^\delta$ to $\range{X,Y}_\star^\delta$ and another bijection from $\range{X,Y}_\star^\delta$ to $\range{X_1|X_2}_\star^\delta$. 
\end{proof}

\begin{proposition} \label{prop:almost_cuv2} For $\delta\in[0,1]$, the following statements hold:
	\begin{itemize}
		\item $\mathfrak{i}^\delta_1\circ X$ and $Y$ are $\delta$-disassociated conditioned on $X$;
		\item $\mathfrak{i}^\delta_2\circ Y$ and $X$ are $\delta$-disassociated conditioned on $Y$.
	\end{itemize}
\end{proposition}

\begin{proof}
	The proof is by \textit{reductio ad absurdum}. 
Assume that $\mathfrak{i}^\delta_1\circ X$ and $Y$ are not $\delta$-disassociated conditioned on $X$. One of the following cases can occur.
\par Case~1: There exist $x\in\range{X}$ and $y_1,y_2\in\range{Y|X(\omega)=x}$ such that $y_1\neq y_2$ and 
$|\range{\mathfrak{i}^\delta_1\circ X|Y(\omega)= y_1,X(\omega) = x} \cap  \range{\mathfrak{i}^\delta_1\circ X|Y(\omega) = y_2,X(\omega) = x}|/|\range{\mathfrak{i}^\delta_1\circ X|X(\omega)=x}|< \delta$. Note that $|\range{\mathfrak{i}^\delta_1\circ X|X(\omega)=x}|=1$ and $\range{\mathfrak{i}^\delta_1\circ X|Y(\omega) = y_1,X(\omega) = x}= \range{\mathfrak{i}^\delta_1\circ X|Y(\omega) = y_2,X(\omega) = x}=\range{\mathfrak{i}^\delta_1\circ X|X(\omega)=x}$. This results in contradiction as it gives  $|\range{\mathfrak{i}^\delta_1\circ X|Y(\omega) = y_1,X(\omega) = x} \cap  \range{\mathfrak{i}^\delta_1\circ X|Y(\omega) = y_2,X(\omega) = x}|/|\range{\mathfrak{i}^\delta_1\circ X|X(\omega)=x}|=1$ while $\delta\leq 1$. 
\par Case~2: There exist $x\in\range{X}$ and $z_1,z_2\in\range{\mathfrak{i}^\delta_1\circ X|X(\omega)=x}$  such that $z_1\neq z_2$ and $|\range{Y|\mathfrak{i}^\delta_1\circ X(\omega) = z_1,X(\omega) = x} \cap  \range{Y|\mathfrak{i}^\delta_1\circ X(\omega) = z_2,X(\omega) = x}|/|\range{Y|X(\omega)=x}|< \delta.$ This is however not possible as $\range{\mathfrak{i}^\delta_1\circ X|X(\omega)=x}$ is a singleton so it cannot contain two distinct points $z_1,z_2$.
\par The proof for showing that $\mathfrak{i}^\delta_2\circ Y$ and $X$ are $\delta$-disassociated conditioned on $Y$ follows the same line of reasoning.
\end{proof}

Propositions~\ref{prop:almost_cuv1} and~\ref{prop:almost_cuv2} show that equivalent uncertain variables $\mathfrak{i}^\delta_1\circ X$ and $ \mathfrak{i}^\delta_2\circ X_2$ are akin to relaxations of the common uncertain variable (c.f., Proposition~\ref{prop:cuv_maximin}). 

\begin{corollary} Let $f:\range{X_1,X_2}\rightarrow \mathbb{R}^m$ be any mapping for which there exists  $g_1:\range{\mathfrak{i}^\delta_1\circ X} \rightarrow \mathbb{R}^m$ or 
	$g_2:\range{\mathfrak{i}^\delta_2\circ Y} \rightarrow \mathbb{R}^m$
	such that $f(X_1,X_2)=g_1\circ \mathfrak{i}^\delta_1\circ X$ or $f(X_1,X_2)=g_2\circ \mathfrak{i}^\delta_2\circ Y$. 
Then, $f$ provides $\log_2(-\delta)$-privacy.
\end{corollary}

\section{Private Multi-Party Function Evaluation}
In this section, we generalize the results of the earlier sections to more than two entities. 
Consider $n\geq 2$ entities, each possessing an uncertain variable $X_i$. Define $X=(X_i)_{i=1}^n$. We are interested in publishing the result of evaluating a function $f:\range{X}\rightarrow \range{Y}$ on the realization of the uncertain variables possessed by those entities. We want this publication  to be done in a privacy-preserving manner. We first generalize the notion of privacy from two-party function evaluations.

\begin{definition}[$\gamma$-Privacy in Multi-Party Function Evaluation] \label{def:multiple_almost_perfect_new} Let $X_1,X_2,\dots,X_n$ be any $n\geq 1$ uncertain variables. For $\gamma\geq 0$, mapping $f:\range{X_1,\dots,X_n}\rightarrow \mathbb{R}^m$ provides $\gamma$-privacy if $\max_{1\leq i\leq n} L_0(f(X_1,\dots,X_n);X_{-i}|X_{i})\leq \gamma$. 
\end{definition}

As earlier said, we might not be able to evaluate any function $f$ in a privacy-preserving manner. Therefore, we might need to approximate function $f$ with one that can be privately evaluated $f':\range{X}\rightarrow \range{Z}$, where $Z=f'\circ X$. Doing so, we publish the result of evaluating another function $f'$ instead $f$.  The error in the function evaluation is measured by 
\begin{align*}
\mathcal{E}(f',f)=\sup_{x\in \range{X}} \rho (f(x),f'(x)),
\end{align*}
where $\rho:\range{Y}\times \range{Z}\rightarrow \mathbb{R}$ is a distance function. In what follows, we use $\rho (y,z)=\|y-z\|$ for $\range{Y},\range{Z}\subseteq \mathbb{R}^{m}$.

\begin{definition}[Accuracy] Any $f'$ is said to be  $\beta$-accurate for $\beta>0$ if $\mathcal{E}(f',f)\leq \beta$. 
\end{definition}

An important problem is to find $\gamma$-private mapping $f'$, if one exists, that is a $\beta$-accurate approximation of a desired function $f$. We investigate this in the remainder of this paper. 

\begin{proposition} \label{prop:1}
	Let $f'$ 
	be such that 
	\begin{align} \label{eqn:2}
	|\range{f'(X)|X_i(\omega)=x_i}|\leq \epsilon+1,\forall (x_i,x_{-i})\in\range{X}, \forall i.
	\end{align}
	Then, $f'$ provides $\log_2(1+\epsilon)$-privacy.
\end{proposition}

\begin{proof}
Note that, since $f'$ is a deterministic mapping, $|\range{f'(X)|X_i(\omega)=x_i,X_{-i}(\omega)=x_{-i}}|=1$. Therefore, if $	|\range{f'(X)|X_i(\omega)=x_i}|\leq \epsilon+1$, we get $	|\range{f'(X)|X_i(\omega)=x_i}|/|\range{f'(X)|X_i(\omega)=x_i,X_{-i}(\omega)=x_{-i}}|\leq \epsilon+1$ for all $(x_i,x_{-i})\in\range{X}$ and all $i$. 
\end{proof}

Linear quantizers can achieve privacy for most functions. These functions have been previously used to provide privacy in the sense of non-stochastic information leakage~\cite{8662687}. Let us first formalize linear quantizers.

\begin{definition}[Linear Quantizer] A $q$-level quantizer $\mathcal{Q}:[x_{\min},x_{\max}]\rightarrow \{b_1,\dots,b_q\}$ is a piecewise constant function  defined as
	\begin{align*}
	Q(x)=
	\begin{cases}
	b_1, & x\in [x_1,x_2),\\
	b_2, & x\in [x_2,x_3),\\
	\hspace{.05in} \vdots & \hspace{.4in}\vdots\\
	b_{q-1}, & x\in[x_{q-1},x_{q}),\\
	b_q, & x\in[x_{q},x_{q+1}],
	\end{cases}
	\end{align*}
	where $(b_i)_{i=1}^q$ are distinct symbols and $x_1\leq x_2\leq \cdots \leq x_q$ are real numbers such that $x_1=x_{\min}$, $x_{q+1}=x_{\max}$, $x_{i+1}-x_{i}=(x_{\max}-x_{\min})/q$ for all $1\leq i\leq q$. This is a mid-point linear quantizer if $b_i=(x_i+x_{i+1})/2$ for all $i$.
\end{definition}

In the next theorem, we show that using linear quantizers provides a universal mechanism for ensuring $\epsilon$-privacy. This is  particularly important as we do not need to develop problem-dependent privacy-preserving policies. 

\begin{theorem} \label{tho:1} Assume that $f$ is Lipschitz continuous, i.e., there exists $L>0$ such that $|f(x)-f(x')|\leq L\|x-x'\|_\infty$ for all $x,x'\in\range{X}$, and $\range{X_i}\subseteq[x_{\min},x_{\max}]$ for all $i$. Then, $f'=\mathcal{M}\circ f$, where $\mathcal{M}$ is a $\lfloor\exp(\gamma)-1\rfloor$-level mid-point linear quantizer over $\range{f(X)}$, is $\gamma$-private and  $\beta$-accurate with $\beta\geq L(x_{\max}-x_{\min})/\lfloor\exp(\gamma)-1\rfloor$.
\end{theorem}

\begin{proof}
Note that $|\range{\mathfrak{M}\circ f(\{x_i\}\times\range{X_{-i}})}|\leq |\range{\mathfrak{M}\circ f(\range{X})}|\leq \lfloor\exp(\gamma)-1\rfloor\leq \exp(\gamma)-1$ because $\range{\mathfrak{M}\circ f(\{x_i\}\times\range{X_{-i}})}\subseteq\range{\mathfrak{M}\circ f(\range{X})}$. In light of Proposition~\ref{prop:1}, this proves that $f'=\mathcal{M}\circ f$ is $\gamma$-private. Due to Lipschitz continuity of $f$,  $f(\range{X})\subseteq [y_{\min},y_{\max}]$, where $y_{\max}-y_{\min}\leq L(x_{\max}-x_{\min})$. Therefore, using $\lfloor\exp(\gamma)-1\rfloor$-level mid-point linear quantizer, $|f(x)-f'(x)|\leq \beta$ with $\beta\geq L(x_{\max}-x_{\min})/\lfloor\exp(\gamma)-1\rfloor$ for all $x,x'\in\range{X}$. This shows that $\mathcal{E}(f',f)\leq \beta$ with $\beta\geq L(x_{\max}-x_{\min})/\lfloor\exp(\gamma)-1\rfloor$. 
\end{proof}

Theorem~\ref{tho:1} shows that private function evaluation can be achieved by uniform quantization of the query responses, where the quantization resolution is a function of privacy budget $\gamma$ and sensitivity of the query $ L(x_{\max}-x_{\min})$ (cf., scale of the Laplace mechanism differential privacy~\cite{10100797835407922841}). Note that $ L(x_{\max}-x_{\min})$ captures the sensitivity of $f$, i.e., how much the output of the function $f$ varies if one of its entries change.  
For the mechanism in Theorem~\ref{tho:1}, we get
\begin{align}
\beta \exp(\gamma)
&\geq \beta (\exp(\gamma)-1)\nonumber\\
&\geq L(x_{\max}-x_{\min})(\exp(\gamma)-1) /\lfloor\exp(\gamma)-1\rfloor\nonumber\\
&\geq L(x_{\max}-x_{\min}).
\end{align}
This inequality provides a utility-privacy trade-off for non-stochastic private function evaluation.

\section{Conclusions}
We consider private function evaluation to provide query responses based on private data of multiple untrusted entities in such a way that no entity can learn something substantially new about the data of others. We prove that uniform quantization of the query responses achieves privacy.

\bibliographystyle{ieeetr}
\bibliography{citation}

\end{document}